\documentclass[a4paper,aps,pra,showpacs,showkeys,preprintnumbers]{revtex4}
\usepackage[utf8]{inputenc}
\usepackage[english]{babel}
\usepackage[T2A]{fontenc}
\usepackage{amssymb,amsfonts,amsmath,mathtext,enumerate,float,dsfont}
\usepackage{graphics,graphicx,epsfig,epstopdf} 
\usepackage{caption}
\usepackage{cmap}    
\usepackage{indentfirst}
\usepackage[usenames]{color}
\usepackage{amsthm}

\renewcommand{\Im}{\mathop{\mathrm{Im}}\nolimits}
\renewcommand{\Re}{\mathop{\mathrm{Re}}\nolimits}
\newtheorem*{theorem}{Theorem}
\newtheorem{corollary}{Corollary}
\newtheorem{exmp}{Example}

\begin{document}
\title{Criteria of minimum squeezing for quantum cluster state generation}
\author{S. B.\, Korolev}
\email{Sergey.Koroleev@gmail.com}
\author{A. D.\, Manukhova}
\author{K. S.\, Tikhonov}
\author{T. Yu.\, Golubeva}
\author{Yu. M.\, Golubev}
\affiliation{St. Petersburg State University, St. Petersburg, 199034 Russia}

\begin{abstract}
In this paper, we assess possibilities of generating cluster states with different topologies being possessed of a finite squeezing resource of the initial oscillators used to generate a cluster state. We obtained the  condition on minimum squeezing required for generating a cluster with a given topology as a simple estimation in terms of the coefficients of the adjacency matrix.
\end{abstract}
\maketitle

\section{Introduction}
 One of the possible ways to create a universal quantum computer is one-way quantum computation  \cite{Raus2}. Its implementation requires a kind of "resource" -- a physical system in a quantum cluster state. Every single element of such a system is connected with one or more others via quantum entanglement so that all together they form a complex physical structure. Here, the entanglement is a  multipartite-entanglement or inseparability in terms of \cite{Furusawa}. Mathematically, such a structure can be described by an undirected graph whose nodes are the elements of the physical system, and the edges are quantum entanglements between them. The type of logic gates, implemented by sequential local measurements over single nodes, depends on the configuration (or topology) of the graph. Since in process of the local measurements quantum entanglement between nodes \textit{reduces}, the cluster state dimension ("the amount of resource" ) gradually decreases.That makes the entire process irreversible or one-way. However, it was shown that such an approach to build a universal quantum computer is effective and in no way inferior to quantum computing on reversible quantum logic gates  \cite{Rausend}.

To date, various schemes of one-way quantum computations have been proposed for both discrete \cite{Raus2} and continuous variables  \cite{Menicucci_1}. Some of them have been successfully realized experimentally \cite{Walther,Kaltenbaek}. 
At that, the main obstacle to their implementation in effective information applications is low degree of scalability as well as in the case of quantum computations on reversible logic gates. At the material level, the scalability problem relates to limitations on the size and topology of the cluster states, that in turn leads to limitations on the number of logical operations and the volume of the processed data. The type of variables that describe the physical system plays a significant role since it determines the nature of the limitations.

In the case of discrete variables, one generates a cluster state on the base of single independent qubits (or qudits \cite{Zhou}). Single-photon sequences from quantum dots \cite{QDot} or single atoms located at the nodes of an optical lattice can be used \cite{Wunderlich}. At that, the probabilistic nature of the operations being performed (e.g. obtaining of single photons, qubit entangling, etc.) restricts the generation of a large-scale cluster state. Due to the low efficiency of these processes, the generation of a large-scale cluster state in practice can take an extremely long time, substantially exceeding the decoherence time of individual qubits. It proves to be  difficult to scale such systems.

For continuous variables, all operations on systems (oscillators) being in the Gaussian quadrature squeezed state that have been used to generate a cluster state are deterministic. For their realization, schemes based on light pulse trains \cite{OPO}, spin waves within atomic ensembles \cite{sways}, and eigenmodes of optomechanical systems \cite{OMS} are proposed. Moreover, methods for generating "hybrid" cluster states based on matter-field oscillators have been offered \cite{sways}. In this case, the limitations for generating a large-scale cluster state relate to a finite degree of quadrature squeezing of the oscillators. By analogy with Duan criterion \cite{Duan}, characterizing the measure of the entanglement of two quantum oscillators, van Loock-Furusawa criterion \cite{Furusawa} for cluster states shows that the entanglement of several ($N\geq2$) quantum oscillators depends directly on their initial degree of squeezing.
Since a large number of squeezed oscillators is required to generate a large-scale cluster state, it is clear that the squeezing degree should be high. Most of the early works on continuous variable cluster state aimed at demonstration of the fundamental possibility to perform one-way quantum computations estimate the degree of squeezing of individual oscillators by its limiting value - infinity. In practice, it turns out to be challenging to get oscillators with a high degree of squeezing. Light pulses with quadrature squeezing of $15$ dB to date have been obtained experimentally  \cite{Vahlbruch}.  At that, the question of the minimum squeezing degree of the quantum oscillator required to generate cluster states of different topologies based thereon has been little studied \cite{HAL}. 

In this paper, we will obtain a criterion determining the minimum degree of quadrature squeezing of the initial oscillators needed to generate a cluster state with a given topology. Using this criterion, we will find which of the cluster state nodes require the highest degree of squeezing for its generation. Then, in the case of  a given degree of squeezing, we estimate the maximum number of adjacent nodes entangled with a selected one. This will enable us to evaluate quantitatively the various quantum cluster topologies and formulate optimal generation strategies, based on these estimates.

\section{Continuous variables cluster state}

Cluster states are a type of quantum multipartite entangled states characterized by an undirected graph. To define the graph $G$  of the cluster state, one need to specify the set of nodes and the set of edges that reveal interconnections between the nodes. Examples are shown in Fig. \ref{Fig}. Every edge connecting $i$-th node with $j$-th one can be characterized by a real number $a_{ij} \in [-1,1]$, called the weight of the edge. The set of these weights defines an adjacency matrix $A$ that, in turn, completely defines the graph $G$. Thus, the graphs in Fig. \ref{Fig} (a) and Fig. \ref{Fig} (b) correspond to the adjacency matrices $A^{(1)}$ and  $A^{(2)}$, respectively:

\begin{align}
A^{(1)}=\begin{pmatrix}
0 & 1 & 0 & 0 \\
1 & 0 & 1 & 0 \\
0 & 1 & 0 & 1\\
0 & 0 & 1 & 0
\end{pmatrix}, \qquad A^{(2)}=\begin{pmatrix}
0 & -1/2 & 0 & 0 & -{1}/{2} & 0 \\
-{1}/{2} & 0 & -{1}/{2} & {1}/{2} & 0 & -{1}/{2} \\
0 & -{1}/{2} & 0 & 0 & {1}/{2} & 0\\
0 & {1}/{2} & 0 & 0 & {1}/{2} & 0\\
-{1}/{2} & 0 & {1}/{2} & {1}/{2} & 0 & {1}/{2}\\
0 & -{1}/{2} & 0 & 0 & {1}/{2} & 0
\end{pmatrix}.
\end{align}
\begin{figure}[H]
\begin{center}
\includegraphics[width=0.75\linewidth]{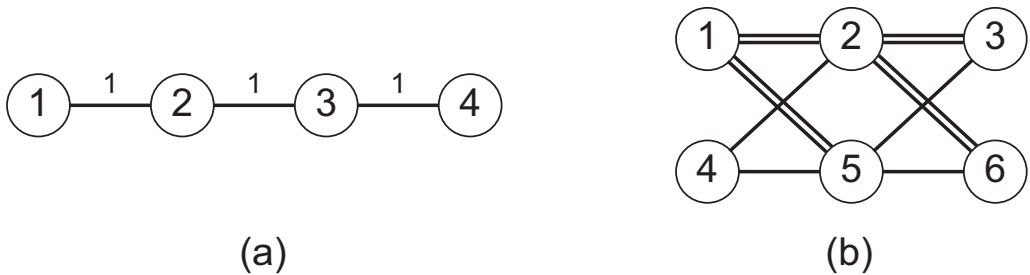}
\end{center}
\caption{Example of two graphs with different topologies. Graph (a): edges connecting adjacent nodes have the same weight equal to $1$. Graph (b): double and single lines are the edges with weights $-1/2$ and $1/2$, respectively.}
\label{Fig}
\end{figure}

A mathematical object -- a graph -- can be put in correspondence with a physical one -- quantum cluster state. As for cluster state generation, the following procedure is usually discussed: one considers $n$ independent quantum harmonic oscillators with squeezed quadratures, every oscillator is assigned by quadrature operators $\hat{x}$  and $\hat{y}$ that obey the canonical commutation relations
 \begin{align}
\left[ \hat{x}_j,\hat{y}_k\right]=\frac{i}{2}\delta _{j,k},
\end{align}
where $j$ and $k$ are numbers of the oscillators, $\delta _{j,k}$ is the Kronecker delta. All the oscillators are assumed to be squeezed in $\hat y$-quadratures  \cite{Gu}, i.e. their variances are less than those of the vacuum state.
\begin{align*} 
\langle \delta \hat{y}_j^2 \rangle < \frac{1}{4}, \qquad j=1,\dots,n.
\end{align*}

One entangles these subsystems in pairs so that coupling force between the $i$-th and $j$-th oscillators would correspond to the element $a_{ij}$ of the adjacency matrix $A$ to satisfy Eqs. (\ref{Null}) and (\ref{def_Null}), see below. It results in the cluster state of the physical system, corresponding to the graph $G$ defined by the adjacency matrix $A$. Note that the nodes of the graph $G$ would correspond to the quantum harmonic oscillators entangled in a given way. Such an entanglement can be described by Bogolyubov's transformation \cite{Bogoljubov} of the initial set of independent quadrature-squeezed oscillators:
 \begin{align} \label{TRB}
\hat{X}_j+i\hat{Y}_j=\sum_{k=1}^n u_{jk}\; \left(\hat{x}_k+i\hat{y}_k \right), \qquad j=1,\dots, n,
\end{align}
where $U=\lbrace u_{jk}\rbrace _{j,k=1}^n$ is a unitary matrix that specifies the set of transformations of the subsystems, so that the result corresponds to the adjacency matrix $A$ (connection of these matrices will be discussed below). $\hat{X}_j$ and $\hat{Y}_j$ are quadrature operators of the $j$-th node of the cluster state.

To describe the quantum statistical properties of discrete variable cluster states, stabilizers are usually used \cite{Raus2}. However, for continuous variables, the most natural way to describe the statistics of a cluster state is to use \textit{nullifier} operators defined for every node of the graph $G$.
\begin{align} \label{Null}
\hat{N}_j=\hat{Y}_j-\sum_{i=1}^n a_{ji}\; \hat{X}_i, \qquad j=1,\dots, n.
\end{align}
By definition, the physical system is in a quantum cluster state if the variances of all of its nullifiers tend to zero in the limit of infinite squeezing of the quantum harmonic oscillators used to generate it \cite{Gu}:
\begin{align} \label{def_Null}
\forall j=1,\dots ,n,\quad \lim \langle \delta \hat{N}_j^2\rangle = 0, \quad \text{при} \quad  \langle \delta\hat{y}^2_1 \rangle \rightarrow 0, \dots , \langle \delta\hat{y}^2_n \rangle \rightarrow 0.
\end{align}

\section{Bogoliubov transformation for cluster states}
To obtain the  desired condition on a minimal squeezing of the initial oscillators required for cluster state generation, we need to rewrite the Bogoliubov transformation $U$ given by the adjacency matrix $A$ in explicit form. Let us derive Eq. (\ref{TRB})  in a vector form
\begin{align} \label{deff_U}
\hat{\vec{X}}+i\hat{\vec{Y}}=U \left( \hat{\vec{x}}+i\hat{\vec{y}}\right)=\left( \Re U \hat{\vec{x}}-\Im U \hat{\vec{y}}\right)+i\left( \Re U \hat{\vec{y}}+\Im U \hat{\vec{x}}\right),
\end{align}
where $\hat{\vec{Y}}=\left(\hat{Y}_1,\hat{Y}_1,\dots, \hat{Y}_n \right)^T$, $\hat{\vec{X}}=\left(\hat{X}_1,\hat{X}_1,\dots, \hat{X}_n \right)^T$ are column vectors formed by the quadratures of the nodes of the cluster state. We introduce vector $\hat{\vec{N}}$, whose elements are the nullifiers (\ref{Null}), by the adjacency matrix $A$ as follows
\begin{align} \label{def_2}
\hat{\vec{N}}=\hat{\vec{Y}}-A \hat{\vec{X}}.
\end{align}
By Eq. (\ref{deff_U}), we express the quadratures of the nodes of the cluster state through the variables
of independent quadrature-squeezed oscillators that were used to generate the cluster, and substitute the result in Eq. (\ref{def_2})
 \begin{equation} \label{Null_1}
\hat{\vec{N}}=\Re U\hat{\vec{y}}+\Im U\hat{\vec{x}}-A\left(\Re U\hat{\vec{x}}-\Im U\hat{\vec{y}}\right)=\left(\Im U-A\cdot \Re U\right)\hat{\vec{x}}+\left(\Re U+A \cdot \Im U\right)\hat{\vec{y}}.
\end{equation}
In order for operators to be nullifiers, their variances have to tend to zero. We assume that all $\hat{y}$-quadratures are squeezed, hence, according to the Heisenberg's uncertainty principle, all the $\hat{x}$-quadratures are stretched. For this reason, the variances of the nullifiers (\ref{Null_1}) tend to zero only if  coefficients in front of the stretched quadratures are zero. Thus, the set of transformations $U$ has to satisfy the equality
\begin{align}
\Im U=A\cdot \Re U.
\end{align}
Hence, the nullifier vector $\hat{\vec{N}}$ and Bogoliubov transformation matrix $U$  are related as following
\begin{align} \label{0}
&\hat{\vec{N}}=\left(I+A^2\right)\Re U\;\vec{y}, \\
&U=\Re U+i\Im U=\left(I+iA\right)\Re U. \label{1}
\end{align}
Here, matrix $\Re U$ remains unknown. In order to specify it, we use the condition of its  unitarity
\begin{align*}
&U^{\dag}U=\left(\left(I+iA\right)\Re U\right)^{\dag}\left(I+iA\right)\Re U=(\Re U)^{T}\left(I+A^2\right)\Re U=I ,
\end{align*}
\begin{align} \label{33}
\left(I+A^2\right)=((\Re U)^{T})^{-1}(\Re U)^{-1}=(\Re U(\Re U)^T)^{-1} \Rightarrow \Re U\left(\Re U\right)^T=\left(I+A^2\right)^{-1}.
\end{align}
The resulting expressions arise from the fact that the matrix $\Re U$ is real, and the matrix $A$ is Hermitian.

 For further analysis we employ polar decomposition of the matrix $\Re U$. A polar decomposition is a representation of an arbitrary square matrix $M$ as a product of a Hermitian $H$ ($H=H^{\dag}$) and a unitary $\mathcal{U}$ ($\mathcal{U}^{\dag}\mathcal{U}=\mathcal{U}\mathcal{U}^{\dag}=I$) matrices, so $M=H\mathcal{U}$. Since the matrix $\Re U$ is real, the polar decomposition for it turns into a product of symmetric matrix $S$ ($S=S^T$) and arbitrary orthogonal matrix  $Q$ ($QQ^T=Q^TQ=I$). Substituting this expansion into Eq. (\ref{33}), we obtain
\begin{align*}
&SQ\left(SQ \right)^T=SQQ^TS=S^2=\left(I+A^2\right)^{-1}  \\
\Rightarrow &S=\left(I+A^2\right)^{-1/2} \\
\Rightarrow &\Re U=\left(I+A^2\right)^{-1/2}Q.
\end{align*}
Thus, Eqs. (\ref{0})  and  (\ref{1})  can be derived as
\begin{align} \label{2}
&\hat{\vec{N}}=\left(I+A^2\right)^{1/2}Q\;\vec{y},\\
&U=(I+iA)(I+A^2)^{-1/2}Q \label{3}.
\end{align}

Let us analyze the result. We have obtained the dependence of the nullifiers and transformation matrix $U$ on the configuration of the graph $G$. The dependences  (\ref{2})--(\ref{3}) contain  an arbitrary orthogonal matrix $Q$ to within these transformations are defined. The question that arises is what this matrix influences. Assuming that one-way quantum computations depend only on a cluster state topology, the authors in  \cite{HAL}  showed the influence of this matrix on the generation procedure. In addition, it can be assumed that via this matrix it is possible to minimize errors that appear in the generation process due to the non-ideality of physical systems.

Thus, for the cluster state characterized by the graph $G$ we have indicated the relation between the explicit form of the Bogoliubov transformation $U$ and the adjacency matrix $A$. Further, basing on
Eqs. (\ref{2}) and (\ref{3}), we will prove the theorem on the relation between the nullifiers' variances and the variances $\langle \delta\hat{y}^2_i\rangle$ in the case of identical initial oscillators used to generate the cluster state.

\section{Citerion of minimal squeezing degree}

To generate a quantum physical system in a cluster state, it is important to know the dependence of nullifier variances on the quadrature variances of independent quantum harmonic oscillators used for cluster generation. In general, if quadratures of the oscillators are squeezed differently, one could not exhibit the explicit formula. However, if oscillators
are squeezed identically, this dependence would be characterized only by the adjacency matrix $A$. Let us prove the following theorem.
\begin{theorem}
Suppose a cluster state that corresponds to the graph $G$  with the adjacency matrix  $A$, given by weight coefficients $a_{ij}$. Let $Q\in M^{n\times n}$ be the orthogonal matrix in a polar decomposition of the real part of the Bogoliubov transformation $U$ $(\Re U=\left(I+A^2\right)^{-1/2}Q)$. If quantum oscillators, used to generate the cluster state, had the same squeezing degree being statistically independent, then the variances of the cluster state nullifiers can be expressed in terms of the variances of the initially squeezed oscillators as follows
 \begin{align}
\langle \delta \hat{N}_j^2\rangle=\left(1+\sum _{i=1}^n a_{ij}^2 \right) \langle\delta\hat{y}^2\rangle , \quad j=1,\dots, n.
\end{align}
\end{theorem}
\begin{proof}
To prove the theorem, let us use Eq.  (\ref{2})
\begin{align}
\hat{\vec{N}}=\left(I+A^2\right)^{1/2}Q\;\hat\vec{y}.
\end{align}
We introduce notation $V=\left(I+A^2\right)^{1/2}$ and derive the $j$-th nullifier as
\begin{align*}
\hat{N}_j=\sum_{k=1}^nv_{jk}\sum _{p=1}^n q_{kp}\:\hat{y}_{p},
\end{align*}
where $v_{jk}$  are matrix elements $V$, $q_{kp}$ are matrix elements $Q$. Let us consider the nullifier variance
\begin{align} \label{sum_1}
\langle \delta \hat{N}_j^2\rangle&=\langle\delta \left( \sum_{k=1}^nv_{jk}\sum _{p=1}^n q_{kp}\:\hat{y}_{p}\right) ^2\rangle=\langle\delta \left( \sum_{p=1}^n\left(\sum _{k=1}^n v_{jk} q_{kp}\:\right)\hat{y}_{p}\right) ^2\rangle= \sum_{p=1}^n\left(\sum _{k=1}^n v_{jk} q_{kp}\right)^2 \langle\delta y^2 \rangle=\nonumber \\
&= \sum_{p=1}^n\left(\sum _{k=1}^n v_{jk} ^2q_{kp}^2+2\sum _{w=1}^n\sum _{r=w+1}^n \left(v_{jw} \;q_{wp}\ \right)\left(v_{jr} \;q_{rp}\ \right)\right) \langle\delta y^2 \rangle. 
\end{align}
Here, we employ statistical independence of the quantum oscillators, i.e. $\langle \delta \hat{y}_i\delta\hat{y}_i\rangle=\delta_{i,j}\langle\delta\hat{y}^2\rangle$, where $\delta_{i,j}$  is Kronecker delta. The first and the second terms of the multiplier on the right-hand side of the Eq. (\ref{sum_1}) can be derived as:
\begin{align} \label{s_1}
\sum_{p=1}^n\sum _{k=1}^n v_{jk} ^2q_{kp}^2=\sum _{k=1}^n \sum_{p=1}^n v_{jk} ^2q_{kp}^2=\sum _{k=1}^n v_{jk} ^2\sum_{p=1}^n q_{kp}^2=\sum _{k=1}^n v_{jk} ^2 \left(\vec{q}_k,\vec{q}_k\right),
\end{align}
\begin{align} \label{s_2}
\sum_{p=1}^n\left(\sum _{w=1}^n\sum _{r=w+1}^n \left(v_{jw} \;q_{wp}\ \right)\left(v_{jr} \;q_{rp}\ \right)\right)=\sum _{w=1}^n\sum _{r=w+1}^n \left(v_{jw} \; v_{jr} \right)\sum_{p=1}^n\left(q_{wp} \;q_{rp}\ \right)=\sum _{w=1}^n\sum _{r=w+1}^n \left(v_{jw} \; v_{jr} \right)\left( \vec{q}_w,\vec{q}_r\right).
\end{align}
$\vec{q}_i=\left(q_{i1},q_{i2},\dots,q_{in} \right)^T$ is a row vector corresponding to the $i$-th row of the matrix $Q$. By definition, an orthogonal matrix is a matrix with columns and rows all  orthonormal,
i.e. $\left( \vec{q}_i,\vec{q}_j\right)=\delta _{i,j}$. Taking into account Eqs. (\ref{s_1}) and (\ref{s_2}), the Eq. (\ref{sum_1}) for nullifier variance can be derived as
\begin{align} 
\langle \delta \hat{N}_j^2\rangle = \sum_{k=1}^n v_{jk} ^2 \langle\delta y^2 \rangle. 
\end{align}

Let us consider the matrix $V=\left(I+A^2\right) ^{1/2}$. Since the adjacency matrix for the cluster state is symmetric ($A=A^T$), it follows that
\begin{align*}
(A^2)^T=(AA)^T=A^TA^T=AA=A^2.
\end{align*}
That means $\left(I+A^2\right)^T=\left(I+A^2\right)$, so  $V$ is also symmetric $(V=V^T)$ since

\begin{align*}
(VV)^T=\left(I+A^2\right)^T=\left(I+A^2\right)=VV.
\end{align*}
Thus, for the diagonal elements $[V^2]_{jj}$, on the one hand, we have
\begin{align*}
[V^2]_{jj}=\sum_{k=1}^n v_{jk} ^2,
\end{align*}
on the other hand,
\begin{align*}
[V^2]_{jj}=[I+A^2]_{jj}=1+[A^2]_{jj}.
\end{align*}
Since $A$ is an adjacency matrix, its symmetry implies that $[A^2]_{jj}=\sum \limits_{i=1}^n a_{ij}^2$. Summarizing the
outcome, we derive the resulting expression for nullifiers
\begin{align} \label{theorem}
\langle \delta \hat{N}_j^2\rangle = \left(1+\sum _{i=1}^n a_{ij}^2 \right) \langle\delta y^2 \rangle, \quad j=1, \dots, n.
\end{align}
\end{proof}
The proven theorem  provides us with a tool for easy calculation of the nullifiers' variances that, in turn, allows one to answer the question of separability of the cluster.
\begin{exmp}
Let us consider the cluster state corresponding to the graph in Fig. \ref{Fig} (b), where double and single lines denote edges with weights $-1/2$ and $1/2$, respectively. According to the theorem proved above, the nullifier variances are equal to
\begin{align*}
\langle \delta \hat{N}_1^2\rangle = \frac{3}{2} \langle\delta y^2 \rangle , \quad \langle \delta \hat{N}_2^2\rangle = 2 \langle\delta y^2 \rangle, \quad \langle \delta \hat{N}_3^2\rangle = \frac{3}{2} \langle\delta y^2 \rangle,\\
\langle \delta \hat{N}_4^2\rangle = \frac{3}{2} \langle\delta y^2 \rangle, \quad \langle \delta \hat{N}_5^2\rangle = 2 \langle\delta y^2 \rangle, \quad \langle \delta \hat{N}_6^2\rangle = \frac{3}{2} \langle\delta y^2 \rangle.
\end{align*}
\end{exmp}
\begin{corollary}
In the case of "unweighted" cluster state, i.e. all its non-zero weight coefficients are equal to one, we have
\begin{align} 
\langle \delta \hat{N}_j^2\rangle = \left(1+\dim \left( Nb\left[j\right] \right) \right) \langle\delta y^2 \rangle, \quad j=1, \dots, n,
\end{align}
where $\dim \left( Nb\left[j\right] \right)$ is the dimension of the set of adjacent nodes (the number of neighbors) for the $j$-th node of the graph $G$.
\end{corollary}
\begin{proof}
From Eq. (\ref{theorem}) it follows that if all non-zero weight coefficients $a_{ij}$ are equal to one, then
the sum on the right-hand side of Eq. (\ref{theorem}) becomes the dimension of the set of adjacent nodes for the $j$-th node ($\dim \left( Nb\left[j\right] \right)$):
\begin{align}
\langle \delta \hat{N}_j^2\rangle =  \left( 1+\dim \left( Nb\left[j\right] \right) \right) \langle\delta y^2 \rangle, \quad j=1, \dots, n.
\end{align}
\end{proof}
Thus, in the case of "unweighted" graph the variance of the $j$-th nullifier is determined only
by the number of nodes of the graph  $G$ connected with the $j$-th node.

\begin{exmp} \label{ex_1}
Let us consider an "unweighted" linear cluster state that corresponds to the $4$-node linear graph in Fig. \ref{Fig} (a). According to Corollary 1, the nullifier variances of this state are equal to
\begin{align*}
\langle \delta \hat{N}_1^2\rangle = 2 \langle\delta y^2 \rangle, \quad
\langle \delta \hat{N}_2^2\rangle = 3 \langle\delta y^2 \rangle, \quad
\langle \delta \hat{N}_3^2\rangle = 3 \langle\delta y^2 \rangle, \quad
\langle \delta \hat{N}_4^2\rangle = 2 \langle\delta y^2 \rangle.
\end{align*}
\end{exmp}

Having proved the theorem on the cluster state nullifier variances in the case of identical independent
quantum oscillators used for its generation, we can analyze  the cluster
via van Loock-Furusawa separability criterion \cite{Furusawa}. This criterion allows one to specify the
minimum squeezing of oscillators required to generate a cluster of a given topology based thereon.
Violation of this condition means that the state is separable.
\begin{corollary}
To generate a cluster state corresponding to the graph $G$ with the adjacency matrix $A$ given by the weight coefficients $a_{ij}$, the squeezing of every initial quantum oscillator has to satisfy the inequality
\begin{align*} 
\langle\delta y^2 \rangle < \min_{(i,j)} \left[\frac{ |a_{ij}|}{2+\sum \limits_{k=1}^n a_{ki}^2+\sum \limits_{l=1}^n a_{lj}^2 } \right], 
\end{align*}
where the minimum on the right-hand side is taken all over the pairs of the adjacent nodes $i$ and $j$.
\begin{proof}  Let us apply van Loock-Furusawa separability criterion to the cluster state corresponding to the graph $G$. In general, this criterion determines the possibility of separating the set of $n$ elements
described by canonical variables $\{\hat{X}_k,\hat{Y}_k\} _{k=1}^n$ by $M$ independent subsets $S_r \left(r=1,\dots ,M \right)$. Mathematically, it can be derived as inequality  \cite{Ukai,me}
\begin{eqnarray} \label{vLF}
&&\langle \delta \hat{b}^2 \rangle+\langle \delta \hat{c} ^2 \rangle \geqslant
\frac{1}{2}\sum_{r=1}^M|\sum_{k\in S_r}\left(h_{k}\tilde g_{k}-\tilde h_{k} g_{k}\right)|,
\end{eqnarray}
where $\hat{b}=\sum \limits _{k=1} ^n\left[h_{k}\hat{X}_k+ g_{k}\hat{Y}_k\right]$, $\hat{c}=\sum \limits _{k=1}^n\left[\tilde h_{k}\hat{X}_k+\tilde g_{k}\hat{Y}_k\right]$  are an auxiliary Hermitian operators that
are linear combinations of all the canonical variables, $h_{k},\tilde h_{k},g_{k},\tilde g_{k}$ are real constants. In our case,
quadratures of the cluster state are taken as canonical variables. We choose the constants
$\lbrace h_{k},\tilde h_{k},g_{k},\tilde g_{k} \rbrace _{k=1}^n$  so that the operators $\hat{b}$ and  $\hat{c}$ turn into nullifiers of adjacent nodes of the cluster state, i.e $\hat{b}=\hat{N}_i=\hat{Y}_i-\sum \limits_{k=1}^n a_{ik}\hat{X}_k$ and $\hat{c}=\hat{N}_j=\hat{Y}_j-\sum \limits_{k=1}^n a_{jk}\hat{X}_k$. Thus, we obtain an inequality
\begin{eqnarray}
\label{Eq27}
&&\langle \delta \hat N_i^2 \rangle+\langle \delta \hat N_j^2 \rangle \geqslant \begin{cases}0, &\quad\text{when} \quad i \in S_r, \quad j \in S_{r}\\
|a_{ij}|, &\quad\text{when} \quad i \in S_r, \quad j \in S_{r'}\end{cases},
\end{eqnarray}
where the nodes $i$ and $j$ of the graph $G$ are connected by an edge with the weight coefficient $a_{ij}$.
The first condition in Eq. (\ref{Eq27}) is trivial since it is always met and carries no information about the connection
between the nodes. The second one is met only if nodes $i$ and $j$ belong to different independent
subsystems. This condition determines the separability criterion for two adjacent nodes of the
cluster state. 
Nodes would be inseparable if inequality is violated, i.e. 
 \begin{eqnarray} \label{Eq28}
&&\langle \delta \hat N_i^2 \rangle+\langle \delta \hat N_j^2 \rangle \textless |a_{ij}|.
\end{eqnarray}

Let us substitute the nullifier variances (\ref{theorem}) in the inequality (\ref{Eq28})
\begin{align*}
\langle \delta \hat{N}_i^2\rangle +\langle \delta \hat{N}_j^2\rangle= \left(2+\sum _{k=1}^n a_{ki}^2+\sum _{l=1}^n a_{lj}^2 \right) \langle\delta y^2 \rangle \textless |a_{ij}|.
\end{align*}
For the variance of the $\hat{y}$-quadrature, we obtain
\begin{align} \label{less}
\langle\delta y^2 \rangle <\frac{ |a_{ij}|}{2+\sum \limits_{k=1}^n a_{ki}^2+\sum \limits_{l=1}^n a_{lj}^2 }.
\end{align}
Depending on the topology of the graph, the right-hand side of this equation may vary. Since the
variances of the nullifiers should tend to zero to generate the cluster state, we minimize the value
on the right-hand side of Eq. (\ref{less}) by the numbers $i$ and $j$  of adjacent nodes of the graph $G$ as
\begin{align} \label{end_2}
\langle\delta y^2 \rangle < \min_{(i,j)} \left[\frac{ |a_{ij}|}{2+\sum \limits_{k=1}^n a_{ki}^2+\sum \limits_{l=1}^n a_{lj}^2 } \right]. 
\end{align}

Thus, we have obtained an estimation criterion on the squeezing of the $\hat{y}$-quadratures of the
initial quantum oscillators used to generate a cluster state.
\end{proof}
\end{corollary}
Let us consider an application of Corollary 2.

\begin{exmp}
Let us consider an "unweighted" \ linear cluster state that corresponds to a 4-node linear graph (Fig. 1a). For this  cluster state, there are two types of conditions
\begin{align}
\langle\delta y^2 \rangle < \frac{1}{5}, \quad \langle\delta y^2 \rangle < \frac{1}{6}.
\end{align}
Since we should choose the minimum, the condition for the cluster to be inseparable is
\begin{align}
\langle\delta y^2 \rangle < \frac{1}{6}.
\end{align}
This condition indicates that to generate the cluster with the given topology it is sufficient to have $4$ oscillators with quadratures squeezed more than $-1,77$ dB.
\end{exmp}
\begin{exmp}
Let us consider the cluster state corresponding to the graph in Fig. \ref{Fig} (b), where double and single lines denote edges with the weights equal to $-1/2$ and $1/2$, respectively. For this state, the inseparability condition is
\begin{align}
\langle\delta y^2 \rangle < \frac{1}{7}.
\end{align}
The required squeezing, in this case, should exceed $-2,43$ dB.
\end{exmp}
The considered examples demonstrate that a finite -- experimentally realizable -- squeezing is sufficient
to generate a cluster state. Note that a more complex and branched structure of the cluster
requires more stringent conditions for its generation. In this respect, the use of simple (e.g. linear)
clusters is preferable if they satisfy computational needs.

The question arises whether such a squeezing is sufficient to implement quantum computations
on the cluster state. The answer is ambiguous. On the one hand, this squeezing would be sufficient
for single quantum gates realization. On the other hand, when considering quantum computations
with a large number of gates, the errors would accumulate and eventually could abolish all the
computation results. It was long believed that an arbitrarily large number of logic gates can
be performed only in a case of infinite squeezing of the oscillators used to generate the cluster states. However, it was shown in \cite{Menicucci}  that any errors can be corrected at
initial squeezing of $20,5$ dB. Unfortunately, at the moment such a squeezing has not yet been
experimentally implemented. Thus, even a low experimentally realizable squeezing can be a
resource for cluster generating and performing of limited quantum computations. In this case, the
"redundant data" procedure \cite{Nielsen} can be used for error correction.

The statement above also allows us to consider the inverse problem: a cluster of what structure
can be generated with a given resource?
\begin{corollary}
When generating an "unweighted" cluster state on the base of statistically independent oscillators with a given squeezing degree $\langle \delta \hat{y}_{fix}^2 \rangle$, the maximum number of neighbors for two adjacent
nodes in the cluster is estimated by an inequality
\begin{align}  \label{end_3}
\max_{(i,j)} \left[\dim \left( Nb\left[j\right] \right)+\dim  \left( Nb\left[i\right]\right) \right]  \textless \frac{1}{\langle \delta \hat{y}_{fix}^2 \rangle}-2,
\end{align}
where the maximum on the left-hand side is taken all over the pairs of adjacent nodes $i$ and $j$.
\begin{proof}
Let us rewrite the inequality (\ref{end_2}) as
\begin{align} 
\frac{2+\max \limits_{(i,j)}  \left(\sum \limits_{k=1}^n a_{ki}^2+\sum \limits_{l=1}^n a_{lj}^2 \right)}{\min \limits_{(i,j)} \left( |a_{ij}|\right)} \textless  \frac{1}{\langle\delta y^2 \rangle} . 
\end{align}
Since we consider an \textit{unweighted} cluster state, we replace all non-zero weight coefficients $a_{ij}$
by $1$. Due to a fixed squeezing for all oscillators we put $\langle\delta y^2 \rangle=\langle\delta y_{fix}^2 \rangle$. That results in the
Eq. (\ref{end_3}).
\end{proof}
\end{corollary}
This corollary shows the maximum number of edges coming from two adjacent nodes of the graph at a given squeezing degree $\langle\delta y_{fix}^2 \rangle$. Let us consider an example of the application of Corollary 3.

\begin{exmp}
Let us suppose that we have oscillators with $\hat{y}$-quadrature squeezing of $6$ dB, i.e $\langle\delta y^2 \rangle\approx 0,06$, and we want to generate an unweighted cluster state. Via Corollary 3 we obtain
a condition on the maximum number of edges of two adjacent nodes in the cluster state
\begin{align}  
\max_{(i,j)} \left[\dim \left( Nb\left[j\right] \right)+\dim  \left( Nb\left[i\right]\right) \right]  \textless 14,
\end{align}
Hence, the maximum number of edges of two adjacent nodes should not exceed $13$. Next, we should distribute these $13$ edges between two nodes so that the designed  cluster state can perform certain quantum computation.
\end{exmp}   
\section{Conclusion}

In this paper we discussed the generation of cluster states basing on the identical independent quadrature-squeezed oscillators. 
We have shown that with the help of the elements of the adjacency matrix the variances of the cluster state nullifiers can be expressed through the variances of the quadratures (or the squeezing degree) of the oscillators. 
This allowed us to formulate a criterion of the minimum degree of quadrature squeezing required to generate a cluster state
with a given topology. With this criterion it was shown that the minimum squeezing degree is determined by the nodes of the cluster  that are connected with the largest number of the adjacent nodes. We estimated the maximum possible number of adjacent nodes on the graph of  a cluster, depending on the squeezing degree of the initial oscillators.

The study of the cluster state topology is interesting from two perspectives: for generating a cluster with a given topology on the basis of the available resource and for organizing computation in limited conditions. The first issue we discussed in this paper, and the second one points the direction for further research.
 
\section*{Acknowledgment}
The reported study was supported by Russian Science Foundation (Project No. 17-72-10171).

\end{document}